\documentclass[a4paper,11pt]{article}

\usepackage{graphicx}

\usepackage{fullpage}
\usepackage{libertine}
\usepackage{color}

\usepackage[ruled,vlined]{algorithm2e}
\SetArgSty{textrm}

\usepackage{amsmath,amsfonts,amssymb,amsthm}

\usepackage[breaklinks=true]{hyperref}
\usepackage[svgnames]{xcolor}
\usepackage[capitalise,nameinlink]{cleveref}
\hypersetup{colorlinks={true},linkcolor={DarkBlue},citecolor=[named]{DarkGreen}}

\usepackage{natbib}

\newtheorem{theorem}{Theorem}[section]
\newtheorem{corollary}[theorem]{Corollary}

\newcommand{\cost}{\text{cost}}

\newcommand{\MAX}{\text{\normalfont \texttt{MAX}}}
\newcommand{\AVG}{\text{\normalfont \texttt{AVG}}}
\newcommand{\costA}{\ensuremath (\AVG \circ \MAX)}
\newcommand{\costB}{\ensuremath (\MAX \circ \AVG)}

\title{\bf Tight Distortion Bounds for Distributed \\ Single-Winner Metric Voting on a Line}
\author{Alexandros A. Voudouris}
\date{School of Computer Science and Electronic Engineering, University of Essex}

\begin{document}

\sloppy
\allowdisplaybreaks

\maketitle

\begin{abstract}
    We consider the distributed single-winner metric voting problem on a line, where agents and alternative are represented by points on the line of real numbers, the agents are partitioned into disjoint districts, and the goal is to choose a single winning alternative in a decentralized manner. In particular, the choice is done by a distributed voting mechanism which first selects a representative alternative for each district of agents and then chooses one of these representatives as the winner. In this paper, we design simple distributed mechanisms that achieve distortion at most $2+\sqrt{5}$ for the average-of-max and the max-of-average social cost objectives, matching the corresponding lower bound shown in previous work for these objectives. 
\end{abstract}

\section{Introduction and Model}
We consider the following voting problem. An instance $I$ consists of a set $N$ of $n$ {\em agents} and a set $A$ of $m$ {\em alternatives}, all of whom are represented by points on the line of real numbers. For any agent $i \in N$ and alternative $x \in A$, let $\delta(i,x)$ be the {\em distance} between $i$ and $x$ on the line (which is equal to the absolute difference between their positions). The agents are also partitioned into a set $D$ of $k$ districts, such that each district contains at least one agent. We denote by $N_d$ the set of agents of district $d\in D$, and by $n_d = |N_d|$ the size of $d$. The goal is to choose an alternative with good social efficiency guarantees using a distributed mechanism that takes as input {\em ordinal information} about the instance, such as the ordering of agents and alternatives on the line, and the ordinal preferences of the agents over the alternatives. To be specific, the {\em preference} of an agent $i$ over the alternatives is a linear ordering of the alternatives such that alternative $x$ is ranked higher than another alternative $y$ if $\delta(i,x) \leq \delta(i,y)$, breaking ties arbitrarily but consistently.

In general, a {\em distributed mechanism} $M$ works as follows:
\begin{itemize}
    \item Step 1: For each district $d\in D$, given the preferences of the agents in $N_d$ over the alternatives and their relative ordering on the line, $M$ decides a {\em representative} alternative $y_d \in A$ for $d$. 
    \item Step 2: Given the representatives of all districts and their relative ordering on the line, $M$ outputs one of them as the overall winner $M(I) \in \bigcup_{d \in D}\{y_d\}$. 
\end{itemize}
Clearly, different distributed mechanisms can be designed by changing the method used for deciding the representative alternatives of the districts or the method for choosing a representative as the overall winner. 

Interpreting the distance $\delta(i,x)$ as the {\em individual cost} of agent $i$ for alternative $x$, we measure the social efficiency of $x$ by some function of the distances of all agents from $x$. There are many well-known social efficiency objectives that have been studied in various social choice problem, such as the {\em social cost} (total or average distance of all agents) and the {\em max cost} (maximum distance among all agents). 
In the context of metric distributed voting, \citet{AFV22} considered social objectives that are composed by some function that is applied over the districts and some function that is applied within the districts. In particular, they focused on the following four objectives. 
\begin{itemize}
\item The {\em average-of-average} cost of $x$ is the average over each district of the average individual cost of the agents therein:
$$(\AVG \circ \AVG)(x) = \frac{1}{k} \sum_{d \in D} \bigg( \frac{1}{n_d} \sum_{i \in N_d} \delta(i,x) \bigg).$$

\item 
The {\em max} cost of $x$ is the max individual cost over all agents: 
$$(\MAX \circ \MAX)(x) = \max_{d \in D} \max_{i \in N_d} \delta(i,x) = \max_{i \in N} \delta(i,x).$$

\item 
The {\em average-of-max} cost of $x$ is the average over each district of the max individual cost therein:
$$(\AVG \circ \MAX)(x) = \frac{1}{k} \sum_{d \in D} \max_{i \in N_d} \delta(i,x).$$

\item
The {\em max-of-average} cost of $x$ is the max over each district of the average individual cost therein:  
$$(\MAX \circ \AVG)(x) = \max_{d \in D} \bigg\{ \frac{1}{n_d} \sum_{i \in N_d} \delta(i,x) \bigg\}.$$
\end{itemize}

We measure the efficiency of a mechanism $M$ with respect to a social objective $f$ (such as the ones defined above) by its {\em distortion}, the worst-case ratio (over all possible instances) of the $f$-value of the alternative chosen by the mechanism over the minimum possible $f$-value among all alternatives: 
\begin{align*}
    \sup_{I=(N,A,D)} \frac{f(M(I))}{\min_{x \in A}f(x)}
\end{align*}
By definition, the distortion of any mechanism is at least $1$. We aim to design distributed mechanisms with an as low distortion as possible. 

\subsection{Our Contribution}
In previous work, \citet{AFV22} considered the problem of distributed metric voting where agents and alternatives are in some arbitrary metric space. By carefully composing centralized voting mechanisms for making decisions within and over the districts, they designed distributed mechanisms with distortion guarantees for general metric spaces and with respect to the aforementioned objectives as well as more general ones. For the special case of a line metric (which is our focus here), they showed that the distortion of distributed mechanisms that use ordinal information is exactly $7$ with respect to the average-of-average cost, exactly $3$ with respect to the max cost, and in the interval $[2+\sqrt{5},5]$ with respect to the average-of-max or the max-of-average cost objectives. 

Inspired by the recent work of \citet{continuous} who designed distributed mechanisms with tight distortion bounds for the continuous distributed facility location problem (where each point on the line of real numbers can be considered as an alternative), we resolve the distortion of distributed mechanisms for voting on a line with respect to the average-of-max and the max-of-average costs. In particular, we design two essentially symmetric distributed mechanism that achieve a tight distortion bound of $2+\sqrt{5}$ for the average-of-max cost and the max-of-average cost. For the average-of-max cost, our distributed mechanism chooses as the favorite alternative of the rightmost agent in each district as the district representative, and then the $(\alpha\cdot k)$-the leftmost alternative as the overall winner. For the max-of-average cost, our distributed mechanism chooses the favorite alternative of the $(\alpha\cdot n_d)$-th agent in each district $d$ as the district representative, and then the rightmost alternative as the overall winner. We show that both mechanisms achieve distortion at most $\max\left\{\frac{3-\alpha}{1-\alpha}, \frac{2}{\alpha}-1\right\}$ for the corresponding objectives; for $\alpha = \frac{3-\sqrt{5}}{2}$, the two terms in the bound balance out to $2+\sqrt{5}$.

\subsection{Other Related Work}
Since its definition by \citet{procaccia2006distortion}, the distortion of voting mechanisms has been studied extensively for several settings under different assumptions about the preferences of the agents. The most well-studied setting is that of single-winner voting that has been considered under the premise of normalized 
 agent valuations~\citep{boutilier2015optimal,ebadian2022optimized,amanatidis2020peeking} as well as metric preferences~\citep{anshelevich2018approximating,gkatzelis2020resolving,kempe2022veto}. Several other models have also been considered, such as multi-winner voting~\citep{caragiannis2017subset,CSV22}, participatory budgeting~\citep{benade2017preference}, and matching~\citep{Aris14,amanatidis2021matching}.

The distortion of distributed voting was first considered by \citep{ratsikas2020distributed} who focused on bounding the distortion of max-weight mechanisms with respect to the social welfare when agents have normalized valuations for the alternatives. \citet{FV21} then studied the distortion of mechanisms for a distributed facility location setting where the agents are positioned on a line and the goal is to choose a single location from a set of alternative ones, which can be discrete (finite) or continuous (infinite). The discrete variant was studied further and generalized to arbitrary metric spaces by \citet{AFV22}, who also introduced and studied the average-of-max and max-of-average costs for the first time. The distortion of distributed mechanisms for the continuous variant on a line was recently resolved by \citet{continuous}. We refer the reader to the recent survey of \citet{distortion-survey} for more details on the distortion of voting mechanisms in different models, and to the survey of \citet{facility-survey} for details on related facility location models.

\section{Average-of-Max Cost} \label{sec:AoM}
Let $\alpha \in [0,1]$ be a parameter. We consider the $\alpha$-Leftmost-of-Rightmost, which works as follows. For each district $d \in D$, the mechanism chooses the favorite alternative of the rightmost agent in $d$ as the representative $y_d$. Afterwards, it chooses the $(\alpha\cdot k)$-th leftmost alternative as the overall winner. See Mechanism~\ref{mech:LoR}. 

\SetCommentSty{mycommfont}
\begin{algorithm}[ht]
\SetNoFillComment
\caption{\sc $\alpha$-Leftmost-of-Rightmost}
\label{mech:LoR}
\For{each district $d\in D$}
{
    $y_d := $ favorite alternative of the rightmost agent in $N_d$\;
}
\Return $w := (\alpha \cdot k)$-th leftmost representative\; 
\end{algorithm}

\begin{theorem}
For average-of-max, the distortion of the {\sc $\alpha$-Leftmost-of-Rightmost} mechanism is at most $\max\left\{\frac{3-\alpha}{1-\alpha}, \frac{2}{\alpha}-1\right\}$.
\end{theorem}

\begin{proof}
Let $w$ be the alternative chosen by the mechanism when given as input an arbitrary instance, and $o$ the optimal alternative. For each district $d$, let $i_d$ be the most distant agent from $w$, and $i_d^*$ the most distant agent from $o$. So, $\costA(w) = \frac{1}{k} \sum_{d \in D} \delta(i_d,w)$, and $\costA(o) = \frac{1}{k} \sum_{d \in D} \delta(i_d^*,o)$. Also, let $\ell_d$ and $r_d$ denote the leftmost and rightmost agents in $d$, respectively. Note that $i_d, i_{d^*} \in \{\ell_d, r_d\}$. We consider the following two cases depending on the relative positions of $w$ and $o$.

\paragraph{Case 1: $o < w$.}
By the triangle inequality, and since $\delta(i_d,o) \leq \delta(i_d^*,o)$, we have
\begin{align}
\costA(w) 
&= \frac{1}{k} \sum_{d \in D} \delta(i_d,w) \nonumber \\
&\leq \frac{1}{k} \sum_{d \in D} \bigg( \delta(i_d,o) + \delta(w,o) \bigg) \nonumber\\
&\leq \costA(o) + \delta(w,o). \label{eq:AoM-1}
\end{align}
Now, let $S$ be the set of representatives that are to the right of $w$. Since $w$ is by definition the $(\alpha \cdot k)$-th leftmost representative, we have that $|S| \geq (1-\alpha)\cdot k$. For every $d$ such that $y_d\in S$, since $o < w \leq y_d$, agent $r_d$ is closer to $w$ than to $o$, and thus $\delta(r_d,o) \geq \frac12 \delta(w,o)$. Hence, 
\begin{align*}
\costA(o) &\geq \frac{1}{k} \sum_{d \in S} \delta(r_d,o) \\
&\geq \frac{1}{k} \cdot |S|\cdot \frac{\delta(w,o)}{2} \\
&\geq \frac{1-\alpha}{2} \cdot \delta(w,o),
\end{align*}
or, equivalently,
\begin{align} \label{eq:AoM-2}
\delta(w,o) \leq \frac{2}{1-\alpha} \cdot \costA(o).
\end{align}
Hence, by \eqref{eq:AoM-1} and \eqref{eq:AoM-2}, we obtain
\begin{align*}
\costA(w) 
&\leq \left( 1+\frac{2}{1-\alpha}\right) \cdot \cost(o)
= \frac{3-\alpha}{1-\alpha} \cdot \costA(o).
\end{align*}

\paragraph{Case 2: $w < o$.}
Let $L$ be a set of $\alpha \cdot k$ districts from the one with the leftmost representative until the one with the $(\alpha\cdot k)$-th representative (which is $w$); denote by $R$ the set of the remaining $(1-\alpha)\cdot k$ districts. 
Observe that:
\begin{itemize}
\item For every $d \in L$, since $y_d$ is the alternative that is closest to $r_d$ and $y_d \leq w < o$, both $r_d$ and $\ell_d$ are closer to $w$ than to $o$. So, $\delta(i_d,w) \leq \delta(i_d^*,o)$. 

\item For every $d \in R$, since $\delta(i_d,o) \leq \delta(i_d^*,o)$ by the definition of $i_d^*$, using the triangle inequality, we obtain $\delta(i_d,w) \leq \delta(i_d,o) + \delta(w,o)\leq \delta(i_d^*,o) + \delta(w,o)$. 
\end{itemize}
Hence, 
\begin{align}
\costA(w) &= \frac{1}{k} \sum_{d \in D} \delta(i_d,w) \nonumber \\
&= \frac{1}{k}\sum_{d \in L} \delta(i_d,w) + \frac{1}{k} \sum_{d \in R} \delta(i_d,w) \nonumber \\
&\leq \frac{1}{k} \sum_{d \in L} \delta(i_d^*,o)
+ \frac{1}{k} \sum_{d \in R} \bigg( \delta(i_d^*,o) + \delta(w,o) \bigg) \nonumber \\
&= \costA(o) + \frac{|R|}{k} \cdot \delta(w,o). \label{eq:AoM-3}
\end{align}
Since $r_d$ is closer to $w$ than to $o$ for every $d \in L$, we also have that
\begin{align*}
\costA(o) \geq \frac{1}{k} \sum_{d \in L} \delta(r_d,o) \geq \frac{|L|}{2k} \delta(w,o),
\end{align*}
or, equivalently, 
\begin{align} \label{eq:AoM-4}
\delta(w,o) \leq \frac{2k}{|L|} \cdot \costA(o).
\end{align}
Therefore, by \eqref{eq:AoM-3} and \eqref{eq:AoM-4}, we obtain
\begin{align*}
\costA(w) \leq \costA(o) + 2\frac{|R|}{|L|} \cdot\costA(o) = 
\bigg(\frac{2}{\alpha} -1 \bigg) \cdot \costA(o).
\end{align*}

Putting everything together, we obtain an upper bound of $\max\left\{\frac{3-\alpha}{1-\alpha}, \frac{2}{\alpha}-1\right\}$.
\end{proof}

Observe that the bound $\max\left\{\frac{3-\alpha}{1-\alpha}, \frac{2}{\alpha}-1\right\}$ consists of two functions of $\alpha$, one that is non-decreasing and one that is non-increasing in $\alpha$. To minimize the maximum between the two, we need to find the value of $\alpha$ for which the two functions intersect. So, we need to solve the equation
\begin{align*}
\frac{3-\alpha}{1-\alpha} = \frac{2}{\alpha}-1 \Leftrightarrow
\alpha^2 - 3\alpha + 1 = 0.
\end{align*}
Since $\alpha < 1$, its solution is $\alpha = \frac{3-\sqrt{5}}{2}$. For this value of $\alpha$, both functions have value $\frac{2}{\frac{3-\sqrt{5}}{2}} - 1  = 2 + \sqrt{5}$, and we obtain the following corollary. 

\begin{corollary}
For average-of-max, the distortion of the {\sc $\frac{3-\sqrt{5}}{2}$-Leftmost-of-Rightmost} mechanism is at most $2+\sqrt{5}$.
\end{corollary}


\section{Max-of-Average Cost} \label{sec:MoA}
Let $\alpha \in [0,1]$ be a parameter. We consider the {\sc Rightmost-of-$\alpha$-Leftmost} mechanism, which works as follows. For each district $d \in D$, the mechanism chooses the favorite alternative of the $(\alpha\cdot n_d)$-th agent in $d$ as the representative $y_d$. Afterwards, it chooses the rightmost alternative as the overall winner. See Mechanism~\ref{mech:RoL}. 

\SetCommentSty{mycommfont}
\begin{algorithm}[ht]
\SetNoFillComment
\caption{\sc Rightmost-of-$\alpha$-Leftmost}
\label{mech:RoL}
\For{each district $d\in D$}
{
    $y_d :=$ favorite alternative of the $(\alpha \cdot n_d)$-th leftmost agent in $N_d$\;
}
\Return $w := $ rightmost representative\; 
\end{algorithm}

\begin{theorem}
For max-of-average, the distortion of the {\sc Rightmost-of-$\alpha$-leftmost} mechanism is at most $\max\left\{\frac{3-\alpha}{1-\alpha}, \frac{2}{\alpha}-1\right\}$. 
\end{theorem}

\begin{proof}
Let $w$ be the alternative chosen be the mechanism when given as input an arbitrary instance, and $o$ the optimal alternative. 
For any district $d$ and alternative $x$, let $\AVG_d(x) = \frac{1}{n_d}\sum_{i \in N_d} \delta(i,x)$ be the total average distance of the agents in $d$ for alternative $x$. So, $\AVG_d(o) \leq \costB(o)$ for every district $d$. Denote by $d^*$ a district that gives the max average cost for $w$, such that $\costB(w) = \AVG_{d^*}(w)$. Also, let $d_w$ be a district represented by $w$. In addition, let $i^*$ and $i_w$ be the $(\alpha\cdot n_d)$-th leftmost agents in districts $d^*$ and $d_w$, respectively. We now switch between the following two cases.

\paragraph{Case 1: $o < w$.} 
By the definition of $d^*$ and the triangle inequality, we have
\begin{align}
\costB(w) 
&= \frac{1}{n_d} \sum_{i \in N_{d^*}} \delta(i,w) \nonumber \\
&\leq \frac{1}{n_d} \sum_{i \in N_{d^*}} \bigg( \delta(i,o) + \delta(o,w) \bigg) \nonumber \\
&\leq \costB(o) + \delta(o,w). \label{eq:MoA-1}
\end{align}
Denote by $S$ the set of agents that are positioned weakly to the right of $i_w$ in $d_w$. 
By the definition of $i_w$, $|S| \geq (1-\alpha)n_{d_w}$. Since $o < w$ and $w$ is the favorite alternative of $i_w$, all agents in $S$ are closer to $w$ than to $o$, and thus $\delta(i,o) \geq \frac12 \delta(w,o)$ for any $i \in S$. Using all these, we obtain:
\begin{align*}
\AVG_{d_w}(o) &= \frac{1}{n_{d_w}}\sum_{i \in N_{d_w}}\delta(i,o) \\
&\geq \frac{1}{n_{d_w}}\sum_{i \in S}\delta(i,o)  \\
&\geq \frac{1}{n_{d_w}} \cdot \frac{|S|}{2} \cdot \delta(w,o) \geq \frac{1-\alpha}{2} \cdot \delta(w,o),
\end{align*}
or, equivalently,
\begin{align}
\delta(w,o) \leq \frac{2}{1-\alpha} \cdot \AVG_{d_w}(o) \leq \frac{2}{1-\alpha} \cdot \costB(o). \label{eq:MoA-2}
\end{align}
Therefore, by \eqref{eq:MoA-1} and \eqref{eq:MoA-2}, we obtain
\begin{align*}
\costB(w) \leq \costB(o) + \frac{2}{1-\alpha} \cdot \costB(o) =  \frac{3-\alpha}{1-\alpha} \cdot \costB(o).
\end{align*}

\paragraph{Case 2: $w < o$.}
Let $L$ be the set of the first $\alpha\cdot n_{d^*}$ agents of $d^*$ (from the leftmost agent to $i^*$), and $R$ be the set of the remaining $(1-\alpha)n_{d^*}$ agents. 
As $w$ is the rightmost representative, $y_{d^*} \leq w < o$. Since $y_{d^*}$ is the favorite alternative of $i^*$, every agent $i \in L$ prefers $w$ over $o$, and thus $\delta(i,w) = \delta(i,o)$.  Using this in combination with the triangle inequality for every agent of $R$, we have
\begin{align}
\costB(w) 
&= \frac{1}{n_{d^*}} \sum_{i \in N_{d^*}} \delta(i,w) \nonumber \\
&= \frac{1}{n_{d^*}} \sum_{i \in L} \delta(i,w) + \frac{1}{n_{d^*}} \sum_{i \in R} \delta(i,w) \nonumber \\
&\leq \frac{1}{n_{d^*}} \sum_{i \in L} \delta(i,o) + \frac{1}{n_{d^*}} \sum_{i \in R} \bigg( \delta(i,o) + \delta(w,o) \bigg) \nonumber \\
&= \frac{1}{n_{d^*}} \sum_{i \in N_{d^*}} \delta(i,o) + \frac{|R|}{n_{d^*}} \cdot \delta(w,o) \nonumber \\
&\leq \costB(o) + \frac{|R|}{n_{d^*}} \cdot \delta(w,o). \label{eq:MoA-3}
\end{align}
Since each agent $i \in L$ prefers $w$ over $o$, we also have that $\delta(i,o) \geq \frac12 \delta(w,o)$, and thus
\begin{align*}
\costB(o) \geq \AVG_{d^*}(o) 
&= \frac{1}{n_{d^*}} \sum_{i \in N_{d^*}} \delta(i,o) \\
&\geq \frac{1}{n_{d^*}} \sum_{i \in N_{d^*}} \delta(i,o) \\ 
&\geq \frac{|L|}{2 n_{d^*}} \cdot \delta(w,o) 
\end{align*}
or, equivalently,
\begin{align}
\delta(w,o) \leq \frac{2n_{d^*}}{|L|} \cdot \costB(o). \label{eq:MoA-4}
\end{align}
Therefore, by \eqref{eq:MoA-3} and \eqref{eq:MoA-4}, we obtain
\begin{align*}
\costB(w) &\leq \costB(o) + 2\frac{|R|}{|L|} \cdot \costB(o) = \left( \frac{2}{\alpha}-1 \right) \cdot \costB(o).
\end{align*}

Putting everything together, we get an upper bound of $\max\left\{\frac{3-\alpha}{1-\alpha}, \frac{2}{\alpha}-1\right\}$.
\end{proof}

By optimizing over $\alpha$, similarly to Section~\ref{sec:AoM}, we obtain the following result.

\begin{corollary}
For max-of-average, the distortion of the {\sc Rightmost-of-$\frac{3-\sqrt{5}}{2}$-Leftmost} mechanism is at most $2+\sqrt{5}$.
\end{corollary}

\section{Open Questions}
In this paper, we showed a tight distortion bound of $2+\sqrt{5}$ with respect to the average-of-max and max-of-average cost functions for the single-winner distributed single-winner voting problem, thus completing the distortion picture with respect to the four basic objectives considered by \citet{AFV22} for the line metric. The most interesting direction for future work is to prove tight distortion bounds for general metric spaces, and also consider other social objectives or settings beyond single-winner voting. 

\bibliographystyle{plainnat}
\bibliography{references}

\end{document}